\documentclass[12pt, draftclsnofoot, onecolumn]{IEEEtran}

\usepackage{cite} 

\ifCLASSINFOpdf
\usepackage[pdftex]{graphicx}
\else
\usepackage[dvips]{graphicx}
\fi

\usepackage[cmex10]{amsmath}
\usepackage{algorithm}
\usepackage{algpseudocode}
\usepackage{amssymb} 
\usepackage{amsthm}
\usepackage{tikz}
\usetikzlibrary{arrows}
\usetikzlibrary{calc,through,backgrounds,shapes}
\usepackage{pgfplots}
\usepackage{graphicx}
\usepackage{caption}
\usepackage{subcaption}
\usepackage{multirow} 

\usepackage{array}
 \usepackage{url}

\newtheorem{lemma}{Lemma}
\newtheorem{corollary}[lemma]{Corollary}
\newtheorem{definition}{Definition}
\newtheorem{proposition}{Proposition}
\newtheorem{theorem}{Theorem}

\makeatletter
\newif\if@borderstar
\def\bordermatrix{\@ifnextchar*{%
 \@borderstartrue\@bordermatrix@i}{\@borderstarfalse\@bordermatrix@i*}%
}
\def\@bordermatrix@i*{\@ifnextchar[{\@bordermatrix@ii}{\@bordermatrix@ii[()]}}
\def\@bordermatrix@ii[#1]#2{%
\begingroup
 \m@th\@tempdima8.75\p@\setbox\z@\vbox{%
 \def\cr{\crcr\noalign{\kern 2\p@\global\let\cr\endline }}%
 \ialign {$##$\hfil\kern 2\p@\kern\@tempdima & \thinspace %
  \hfil $##$\hfil && \quad\hfil $##$\hfil\crcr\omit\strut %
  \hfil\crcr\noalign{\kern -\baselineskip}#2\crcr\omit %
  \strut\cr}}%
 \setbox\tw@\vbox{\unvcopy\z@\global\setbox\@ne\lastbox}%
 \setbox\tw@\hbox{\unhbox\@ne\unskip\global\setbox\@ne\lastbox}%
 \setbox\tw@\hbox{%
  $\kern\wd\@ne\kern -\@tempdima\left\@firstoftwo#1%
  \if@borderstar\kern 2pt\else\kern -\wd\@ne\fi%
 \global\setbox\@ne\vbox{\box\@ne\if@borderstar\else\kern 2\p@\fi}%
 \vcenter{\if@borderstar\else\kern -\ht\@ne\fi%
  \unvbox\z@\kern -\if@borderstar2\fi\baselineskip}%
 \if@borderstar\kern-2\@tempdima\kern2\p@\else\,\fi\right\@secondoftwo#1 $%
 }\null \;\vbox{\kern\ht\@ne\box\tw@}%
\endgroup
}
\makeatother

\title{Characterization of SINR Region for Multiple Interfering Multicast in Power-Controlled Systems}
\author{Yi Chen and Chi Wan Sung \thanks{Yi Chen is with Department of
    Information Engineering, The Chinese University of Hong
    Kong. Email: chelseachenyi@gmail.com. 
  Chi Wan Sung is with Department of Electronic Engineering,
  City University of Hong Kong. Email: albert.sung@cityu.edu.hk. This work was supported in part by a grant from the University Grants Committee of the Hong Kong Special Administrative Region, China under Project AoE/E-02/08, and in part by a grant from the Research Grants Council of the Hong Kong Special Administrative Region, China under Project CityU 121713.}}

\begin{document} 
 \maketitle
 \pagestyle{empty}
\thispagestyle{empty}
\begin{abstract}
This paper considers a wireless communication network consisting of
multiple interfering multicast sessions. Different from a unicast
system where each transmitter has only one receiver, in a multicast
system, each transmitter has multiple receivers. It is a well known result for wireless unicast systems that the feasibility of an
signal-to-interference-plus-noise power ratio (SINR) without power constraint is decided by the
Perron-Frobenius eigenvalue of a nonnegative matrix. We generalize
this result and propose necessary and sufficient conditions for the 
feasibility of an SINR in a wireless multicast system with and without
power constraint. The
feasible SINR region as well as its geometric properties are studied. Besides, an iterative algorithm is proposed which can
efficiently check the feasibility condition and compute the boundary points of the feasible SINR region.
\end{abstract}


\begin{keywords}
Wireless multicast system, power control,
  signal-to-interference-plus-noise power ratio (SINR), SINR
  feasibility, SINR region, Perron-Frobenius eigenvalue.
\end{keywords}


\section{Introduction}
In wireless communication systems, interference is an inherent
phenomenon. Due to the broadcast nature of wireless channels,
interference arises whenever multiple transmitter-receiver pairs are
active concurrently in the same frequency band, and each receiver is
only interested in retrieving information from its own
transmitter. For a particular receiver, the received signal is a
superposition of its desired signal, interfering signals and background
noise. SINR, defined
as the power of desired signal divided by the sum of the power of
interfering signals and the power of noise, is a widely used
performance measure for wireless communication systems. It is
analogous to signal-to-noise ratio (SNR) used for single
user communication, which has clearly understood implication on the bit error rate (BER) and capacity
for additive white Gaussian noise (AWGN) channels. 
Using SINR as a surrogate for BER and capacity implicitly assumes that
the
interference is an AWGN. Although there
are limitations of this assumption, as reported in \cite{5910125,6373675}, the importance of
SINR has never been doubted. 

For a system consisting of multiple point-to-point communication
sessions, also referred to as \emph{unicast system}, the SINRs of all receivers form a
vector. The feasible SINR region includes all the SINR vectors that can be
achieved by some transmission powers. 
The geometric properties of feasible SINR region has been
studied in \cite{1010870,4012511,5730591}. Reference
\cite{1010870} proves that in the case of unlimited transmission power, the feasible
SINR region is log-convex.  In \cite{4012511}, it is shown that under
a total power constraint, the infeasible SINR region is not convex. Reference
\cite{5730591} considers a system with only three transmitter-receiver
pairs without power constraint, and
shows that the feasible SINR region is concave. It also provides
certain technical conditions under which a concavity result for
systems with a general number of users is established. In
\cite{5466510}, for the cases that the transmission powers are subject
to arbitrary linear constraints, a mathematical expression for the
boundary points of the SINR region is obtained.

In this paper, we consider the feasible SINR region for systems
consisting of multiple point-to-multipoint communication sessions,
also referred to as \emph{multicast system}.  
Multicast enables data to be delivered from a source node to multiple
destination nodes. Practical examples of such configurations include
cellular networks and two-way relay networks. In cellular networks, a
base station multicasts a file to multiple mobile devices that request
the file at the same time \cite{5285177}. In
two-way relay networks, when network coding is applied, a relay
multicasts the coded packets to two sink nodes \cite{7024873}.
The power control and scheduling for wireless multicast systems have
been studied in \cite{1204269,4606845,9384025,6497024} and the
references therein. All these works aim to either minimize the system
power or maximize the system throughput, subject to the constraint
that the SINR of all receivers are larger than a given threshold. The
feasibilities of the problems, however, are unknown. 

For a wireless unicast system, the feasibility of an SINR vector
without power constraint is determined by the Perron-Frobenius
eigenvalue of a nonnegative matrix \cite{438920,120145}. In this
paper, we generalize this result to a wireless multicast system. We
first propose a necessary and sufficient condition under which an SINR
is achievable without power limitation. Based on this
condition, we figure out the feasible SINR region by giving its
boundary points. The approach is to find the
farthest point of the feasible SINR region from the origin in a given
direction. An iterative algorithm is proposed to find the farthest point,
which is also a distributed power control algorithm to solve the power
balancing problem \cite{438920} aimed to maximize the minimal SINR of
all receivers. 

Then we analyse the geometric properties of the feasible and
infeasible SINR regions. It is found that the feasible SINR region of a multicast system is in fact
the intersection of the feasible SINR regions of all its embedded
unicast systems. Based on the results in
\cite{1010870,4012511,5730591} for unicast systems, we
show that the feasible SINR region of a multicast system is
log-convex, and the infeasible SINR region of a multicast system with
two multicast sessions is convex. We also show by an example that, the
convexity property of the infeasible SINR region does not hold for a
general multicast system with more than two multicast sessions.

Later, the necessary and sufficient condition for the feasibility of
an SINR in a multicast system is extended to include linear
constraints on the power. This result generalizes the results in
\cite{5466510} for unicast systems. Besides, in \cite{5466510}, the
zero-outage SINR region for a time-varying system is also considered, where the channel gains are
selected from a finite set. Suppose the transmission powers are not allowed to vary with the
channel gains, we establish a reduction that maps any instance of the
zero-outage SINR problem in a time-varying unicast system to a
corresponding instance of the feasible SINR problem in a
time-invariant multicast system. The idea is to regard the multiple receivers in a multicast session
as an identical receiver that can experience a finite set of channel
gains.

The rest of the paper is organized as follows: In Section
\ref{section:system}, the system model and problem formulation are
presented. The necessary and sufficient condition on the feasibility
of an SINR vector is provided in Section \ref{section:feasiblility}. Section
\ref{section:algorithm} gives the characterizations of the SINR
region and proposes an iterative algorithm. Section
\ref{section:geometric} studies the geometric properties of the
feasible SINR region. Section \ref{section:constraint} extends the
study to include power constraints. Finally, the paper is concluded in
Section \ref{section:conclusion}.

\emph{Notation}: The following notations are used throughout this
paper. Vectors are denoted in bold small letter, e.g., $\mathbf x$, with their
$i$th entry denoted by $x_i$. They are regarded as column vectors
unless stated otherwise. Matrices are denoted by bold capitalized
letters, e.g., $\mathbf X$, with $X_{ij}$ denoting the $(i, j)$th
entry. Vector and matrix inequalities are component-wise inequalities, e.g.,
$\mathbf x \geq \mathbf y$ if $x_i\geq y_i$ for all $i$; $\mathbf
X\geq \mathbf Y$ if $X_{ij}\geq Y_{ij}$ for all $i$ and $j$. The
cardinality of a set is denoted by ``$|\cdot|$". The Euclidean norm of
a vector is
denote by ``$||\cdot||$" . The transpose of a
vector or matrix is denoted by $(\cdot)^{T}$. $\mathbf I$ represents
an identity matrix with compatible size. $\mathbf 0$ represents a
vector with compatible size whose entries are all zero.

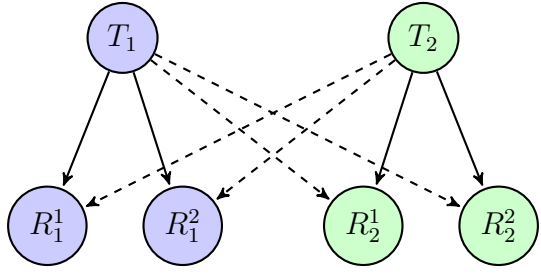
\begin{figure}
\centering
\begin{tikzpicture}[->,>=stealth',shorten >=1pt,node distance=3cm,
  thick, node/.style={circle,minimum size = 13mm,fill=blue!20,draw}]
\node at (1,3)  [circle,fill=blue!20,draw,] (1){$T_1$};
\node at (5,3)  [circle,fill=green!20,draw] (2){$T_2$};
\node at (0,.5)  [circle,fill=blue!20,draw] (3){$R_1^1$};
\node at (1.8,0.5)  [circle,fill=blue!20,draw] (4){$R_1^2$};
\node at (4.2,0.5)  [circle,fill=green!20, draw] (5){$R_2^1$};
\node at (6,0.5)  [circle,fill=green!20, draw] (6){$R_2^2$};
\path
  (1)   edge (3)
      edge (4)
      edge[dashed] (5)
       edge[dashed] (6)
 (2)   edge (5)
      edge (6)
      edge[dashed] (3)
       edge[dashed] (4);
\end{tikzpicture}
\caption{Example of a multicast network consisting of two multicast
  sessions. Transmitter $T_1$ wants to transmit data to both $R_1^1$
  and $R_1^2$. Transmitter $T_2$ wants to transmit data to both $R_2^1$
  and $R_2^2$. Their transmitted signals interfere with each
  other. The solid lines represent intended links and the dashed lines
represent interfering links.}
\label{fig:example}
\end{figure}

\section{System Model}\label{section:system}
Consider a general wireless communication network consisting of $N$ multicast sessions. The $N$
transmitters are denoted by $T_i$ for $i=1,\ldots, N$. Each $T_i$
wants to multicast common data packets to $K_i$ receivers, denoted by $R_{i}^{k_i}$ for
$k_i=1,\ldots,K_i$. Without loss of generality assume $K_i\geq 1$. If
$K_i = 1$ for all $i$, the scenario reduces to the unicast case. The total number
of receivers in the system is $K = \sum_i^N K_i$. Define $\mathcal K_i
= \{1,2,\ldots,K_i\}$ for $i = 1,\dots,N$. Fig. \ref{fig:example}
illustrates an example of such network. Let $p_i$ be the transmission power of transmitter $T_i$ and $\mathbf p = [p_1,\ldots,p_N]^T$. The channel gain between $T_j$ and $R_{i}^{k_i}$ is denoted by $g_{r_i^{k_i},t_j}$. 
All the multicast sessions share the same channel and thus interfere with each other. We assume that interference caused by simultaneous transmissions is treated as AWGN with variance identical to the received power. 
The SINR of receiver $R_i^{k_i}$ is given by
\begin{equation}
\gamma_i^{k_i}(\mathbf p) = \frac{g_{r_i^{k_i},t_i} p_i}{\sum_{j\neq i}g_{r_i^{k_i},t_j} p_j + \sigma^2},
\end{equation}
where $\sigma^2$ is the variance of the background noise and without
loss of generality, it is
assumed to be identical for all receivers. We define the SINR of the $i$-th multicast session as
\begin{equation*}
\gamma_i (\mathbf p)= \min_{k_i\in \mathcal K_i} \left\{\gamma_i^{k_i}(\mathbf p)\right\}.
\end{equation*}
The SINR vector of the system is 
\begin{equation*}
 \Gamma (\mathbf p)= [\gamma_1(\mathbf p),\gamma_2(\mathbf
 p),\ldots,\gamma_N(\mathbf p)].
 \end{equation*}
In this paper, we analyze the feasible SINR region of a multicast system, that is, 
\begin{equation*}
\Upsilon = \left\{\Gamma(\mathbf p) \in \mathbb R^N: \mathbf p\geq
  \mathbf 0, \mathbf p \in \mathbb R^N\right\}.
\end{equation*}

\begin{proposition}\label{pro:aofis}
Given a vector $\boldsymbol\mu\in \mathbb R^N$. There exists a power vector
$\mathbf p^*\geq \mathbf 0$ such that $\Gamma(\mathbf p^*) = \boldsymbol\mu$,
if and only if there exits a power vector $\mathbf p'\geq \mathbf 0$
such that $\Gamma(\mathbf p') \geq \boldsymbol\mu$.
\end{proposition}
\begin{proof}
The ``only if'' part is trivial and we show the ``if'' part. Suppose
$\gamma_i(\mathbf p')>\mu_i$ for some $i$.
Fix such an $i$. Since $\gamma_i(\mathbf p)$ is monotonically
decreasing as $p_i$ is decreasing, we can find a $0<p_i^{(1)}<p_i'$ and
let $\mathbf p^{(1)} = [p_1',\ldots,p_{i-1}', p_i^{(1)},p_{i+1}',\ldots,p_N']^T$ such that
$\gamma_i(\mathbf p^{(1)})=\mu_i$. On the other hand, since
$\gamma_j(\mathbf p)$ for $j\neq i$ is monotonically
increasing as $p_i$ is decreasing, we have $\gamma_j(\mathbf
p^{(1)})\geq \mu_j$. By keeping on decreasing the power of
transmitters
that achieve higer SINR than $\boldsymbol\mu$, we obtain a sequence
$p_i^{(1)},p_i^{(2)},\ldots,p_i^{(t)},\ldots$ for each $i=1,\ldots,N$. It can be
seen that these sequences are monotonically decreasing and lower
bounded by zero, so they are convergent. Denote the limit point by
$\mathbf p^*$. For any arbitrarily small $\delta>0$ and for all $i$,
since $\gamma_i(\mathbf p)$ is continuous with $\mathbf p$, there
exists a sufficiently large $T$, when $t>T$, $|\gamma_i(\mathbf p^*) -
\gamma_i(\mathbf p^{(t)})| <\delta$. Meanwhile, since
$\gamma_i(\mathbf p^{(t')})=\mu_i$ for some $t'\geq T$, we have
$|\gamma_i(\mathbf p^*) - \mu_i| <\delta$. Therefore $\gamma_i(\mathbf
p^*) =\mu_i$ for all $i$.
\end{proof}

By Proposition \ref{pro:aofis}, we say that an SINR vector $\boldsymbol\mu = [\mu_1,\ldots,\mu_N]$ is \emph{feasible} if and only if there exists a power vector $\mathbf p\geq \mathbf 0$ such that
\begin{equation}\label{equ:sdiod}
p_i -   \sum_{j\neq i} \mu_i
\frac{g_{r_i^{k_i},t_j}}{g_{r_i^{k_i},t_i}}p_j \geq \mu_i \frac{\sigma^2}{g_{r_i^{k_i},t_i}}, \forall k_i\in \mathcal K_i,\forall i.
\end{equation}
In matrix form, it is
\begin{equation}\label{equ:kfkk}
\mathbf A(\boldsymbol\mu) \mathbf p\geq \mathbf n(\boldsymbol\mu),
\end{equation} 
where
\begin{align*}
\mathbf A(\boldsymbol\mu) = \left[
\begin{smallmatrix}
1 & -\mu_1\frac{g_{r_1^{1},t_2}}{g_{r_1^{1},t_1}} & \cdots &  -\mu_1\frac{g_{r_1^{1},t_N}}{g_{r_1^{1},t_1}}\\
1 & -\mu_1\frac{g_{r_1^{2},t_2}}{g_{r_1^{2},t_1}} & \cdots &  -\mu_1\frac{g_{r_1^{2},t_N}}{g_{r_1^{2},t_1}}\\
 \vdots& \vdots& \ddots &\vdots  \\
1 & -\mu_1\frac{g_{r_1^{K_1},t_2}}{g_{r_1^{K_1},t_1}} & \cdots &  -\mu_1\frac{g_{r_1^{K_1},t_N}}{g_{r_1^{K_1},t_1}}\\
 -\mu_2\frac{g_{r_2^{1},t_1}}{g_{r_2^{1},t_2}}&1  & \cdots &  -\mu_2\frac{g_{r_2^{1},t_N}}{g_{r_2^{1},t_2}}\\
 \vdots& \vdots& \ddots &\vdots  \\
 -\mu_2\frac{g_{r_2^{K_2},t_1}}{g_{r_2^{K_2},t_2}}&1  & \cdots &  -\mu_2\frac{g_{r_2^{K_2},t_N}}{g_{r_2^{K_2},t_2}}\\
 \vdots& \vdots& \ddots &\vdots  \\
-\mu_N\frac{g_{r_N^{K_N},t_1}}{g_{r_N^{K_N},t_N}}&-\mu_N\frac{g_{r_N^{K_N},t_2}}{g_{r_N^{K_N},t_N}}  & \cdots &  1
\end{smallmatrix}\right]
=
\begin{bmatrix}
  \mathbf a_{1}^{1}\\
\mathbf a_{1}^{2}\\
\vdots \\
\mathbf a_{1}^{K_1} \\
\mathbf a_{2}^{1}\\
\vdots \\
\mathbf a_{2}^{K_2}\\
\vdots\\
\mathbf a_{N}^{K_N}
\end{bmatrix}
\in \mathbb R^{K\times N} 
\end{align*}
and
\begin{align*}
\mathbf n (\boldsymbol\mu)&=
\Big[\overbrace{\frac{\mu_1 \sigma^2}{g_{r_1^{1},t_1}},\frac{\mu_1 \sigma^2}{g_{r_1^{2},t_1}}\ldots,\frac{\mu_1 \sigma^2}{g_{r_1^{K_1},t_1}}}^{K_1},
\ldots,\overbrace{\frac{\mu_N \sigma^2}{g_{r_N^{1},t_N}},\ldots,\frac{\mu_N \sigma^2}{g_{r_N^{K_N},t_N}}}^{K_N}\Big]^{T}\\
& = [n_1^1,n_1^2\ldots,n_1^{K_1},
\ldots,n_N^{1},\ldots,n_N^{K_N}]^{T}\in \mathbb R^{K\times 1}.
\end{align*}
Each row of $\mathbf A (\boldsymbol\mu)$ corresponds to a receiver. For the
convenience of discussion, we use $\mathbf a_{i}^{k_i} \in
\mathbb R^N$ to denote the row of $\mathbf A (\boldsymbol\mu)$ that corresponds to receiver
$R_{i}^{k_i}$. As in the form \eqref{equ:kfkk}, the feasibility of $\boldsymbol\mu$ can be
checked through linear programming \cite{6554630}. 
However, in a different way, we propose a necessary
and sufficient condition on the feasibility, which generalizes the
Perron-Frobenius eigenvalue criteria for unicast systems (square
matrices). This condition is used to explicitly characterize the
feasible SINR region $\Upsilon$ and to prove some geometric properties of it. Before further discussion, we given some definitions.

Define set
\begin{equation*}
\mathcal G(\boldsymbol\mu) = \Big\{\mathbf G=
\begin{bmatrix}
\mathbf a_1^{k_1}\\
\mathbf a_2^{k_2}\\
\vdots\\
\mathbf a_N^{k_N}
\end{bmatrix}\in \mathbb R^{N\times N}:
k_i \in \mathcal K_i\ \text{for}\ i=1,\dots,N
\Big\}.
\end{equation*}
Notice that, for each $\mathbf G \in \mathcal G(\boldsymbol\mu)$, only one receiver is
involved for each transmitter, which is a unicast scenario. So
$\mathcal G (\boldsymbol\mu)$ is the set including all the embedded unicast systems and its
size is $\prod_{i=1}^{N} K_i$. 
Considering the example in Fig. \ref{fig:example}, the four embedded
unicast systems are
\begin{equation*}
\mathcal G (\boldsymbol\mu) = \left\{ 
\begin{bmatrix}
\mathbf a_1^{1}\\
\mathbf a_2^{1}
\end{bmatrix}, 
\begin{bmatrix}
\mathbf a_1^{2}\\
\mathbf a_2^{1}
\end{bmatrix}, 
\begin{bmatrix}
\mathbf a_1^{1}\\
\mathbf a_2^{2}
\end{bmatrix}, 
\begin{bmatrix}
\mathbf a_1^{2}\\
\mathbf a_2^{2}
\end{bmatrix}
 \right\}.
\end{equation*}
In subsequent discussion, we also use $\mathbf k
= (k_1,k_2,\ldots,k_N)$ to specify a $\mathbf G \in \mathcal
G(\boldsymbol\mu)$. Let $\mathbf n_{\mathbf
  G}=[n_1^{k_1},n_2^{k_2},\ldots,n_N^{k_N}]^{T}$ denote the noise vector with entries of $\mathbf n (\boldsymbol\mu)$ that correspond to the
 receivers in $\mathbf G$. For the simplicity of notation, we sometimes drop the argument
$\boldsymbol\mu$ of $\mathbf A$, $\mathbf n$ and $\mathcal G$ when the context is clear.

\begin{definition}{\textnormal{\cite{92048}}}
A matrix $\mathbf X$ is called \emph{nonnegative} if $\mathbf X\geq
\mathbf 0$. A nonnegative square matrix $\mathbf X$ is \emph{irreducible} if for every pair $(i,j)$ of its index set, there exists a positive integer $n\equiv n(i,j)$ such that $X_{ij}^{(n)}>0$, where $X_{ij}^{(n)}$ is the $(i,j)$th entry of $\mathbf X^n$.
\end{definition}

\begin{definition}{\textnormal{\cite{92048}}}
Let $\mathbf X$ be an irreducible nonnegative square matrix. The \emph{Perron-Frobenius eigenvalue} of $\mathbf X$ is the maximum of the absolute value of eigenvalues of $\mathbf X$, and is denoted by $\lambda(\mathbf X)$.
\end{definition}

Let $\mathbf 1 = [1,\ldots,1]$ be a vector whose components are all
$1$. For each $\mathbf G\in \mathcal G(\mathbf 1)$, $\mathbf I-\mathbf
G$ is the normalized interference link gain matrix of the corresponding embedded
unicast system.
\begin{definition}
A multicast system is called irreducible if and only if the matrices $\mathbf I-\mathbf
G$ for $\mathbf G\in \mathcal G(\mathbf 1)$ are all irreducible.
\end{definition}
It needs to be mentioned that if a multicast system is irreducible, as long as
$\boldsymbol\mu>\mathbf 0$, the matrices $\mathbf I-\mathbf
G$ for $\mathbf G\in \mathcal G(\boldsymbol\mu)$ are all
irreducible. Throughout the paper, we assume that the multicast system is irreducible.

\section{Feasibility Condition for SINR}\label{section:feasiblility}
In this section, we give a necessary and sufficient condition for
the feasibility of an SINR vector in a wireless multicast system. Recall that
in a wireless unicast system, the
following theorem from \cite{438920,281084} is the fundamental results that characterize the feasibility.

\begin{theorem}\textnormal{\cite{281084}}\label{theor:fjow} 
Consider a unicast network setting $\mathbf G$ and assume $\mathbf
I-\mathbf G$ is irreducible. The following statements are equivalent:
\begin{enumerate}
\item There exists a power vector $\mathbf p \geq \mathbf 0$ such that
    $\mathbf G \mathbf p\geq \mathbf 0$.
\item $\lambda( \mathbf I-\mathbf G) < 1$.
\item $\mathbf G^{-1} = \sum_{k=0}^{\infty} ( \mathbf I-\mathbf G)^k$
  exists and is positive component-wise, with $\lim_{k\rightarrow
    \infty} ( \mathbf I-\mathbf G)^k = 0$.
\end{enumerate}
Moreover, there exists $\mathbf p\geq\mathbf 0$ such that $\mathbf G \mathbf p = \mathbf 0$ if and only if $\lambda( \mathbf I-\mathbf G) =1$.
\end{theorem}

For a multicast system, the main result of this paper is the following theorem.
\begin{theorem}\label{theor:fjiwo}
 Consider a multicast network setting $\mathbf A(\boldsymbol\mu)$ and
 assume it is irreducible, i.e., the matrices $\mathbf I-\mathbf
 G$ for $\mathbf G\in \mathcal G(\boldsymbol\mu)$ are all irreducible. There exists a power vector $\mathbf p \geq \mathbf 0$ such that $\mathbf A(\boldsymbol\mu)
 \mathbf p\geq \mathbf n(\boldsymbol\mu)$ if and only if 
$\max_{\mathbf G\in
   \mathcal G(\boldsymbol\mu)} \{\lambda(\mathbf I-\mathbf G) \}< 1$.
\end{theorem}

Theorem \ref{theor:fjiwo} basically says that for a wireless multcast system, an SINR vector $\boldsymbol \mu$  is
feasible if and only if $\boldsymbol \mu$ is feasible to any of its embedded
unicast system. 

\begin{corollary}\label{corol}
When there are only two multcast sessions, i.e., $N=2$, the feasibility of $\boldsymbol\mu$ is determined by the unicast system specified by
\begin{equation*}
\mathbf G^* = \begin{bmatrix}
\mathbf a_1^{k_1^*}\\ \mathbf a_2^{k_2^*} \end{bmatrix}\ \text{where}\ k_i^* =\arg \max_{k_i\in \mathcal K_i}\{ \mu_i\frac{g_{r_i^{k_i},t_j}}{g_{r_i^{k_i},t_i}} \}, i=1,2, j\neq i.
\end{equation*} That is,
$\boldsymbol\mu$ is feasible if and only if $\lambda(\mathbf I-\mathbf G^*) < 1$.
\end{corollary}
Corollary \ref{corol} follows straightforwardly from Theorem \ref{theor:fjiwo}. Note that when $N=2$, for any $\mathbf G\in \mathcal G$, we have
\begin{equation}
\lambda(\mathbf I-\mathbf G) = \sqrt{\mu_1
\frac{g_{r_1^{k_1},t_2}}{g_{r_1^{k_1},t_1}}\times \mu_2
\frac{g_{r_2^{k_2},t_1}}{g_{r_2^{k_2},t_2}}}.
\end{equation}
So $\max_{\mathbf G\in
   \mathcal G} \{\lambda(\mathbf I-\mathbf G) \} =\lambda(\mathbf I-\mathbf G^*)$.

The proof of the necessary condition for Theorem \ref{theor:fjiwo} is straightforward. Suppose there exists
a power vector $\mathbf p \geq \mathbf 0$ such that $\mathbf A\mathbf p\geq
\mathbf n$. Then for any $\mathbf G \in \mathcal G$, we have $\mathbf G\mathbf
p \geq \mathbf n_{\mathbf G}\geq \mathbf 0$. By Theorem \ref{theor:fjow}, $\lambda(\mathbf I-\mathbf G) < 1$ for all $\mathbf G$, which implies $\max_{\mathbf G\in
   \mathcal G} \{\lambda(\mathbf I-\mathbf G) \}< 1$.

In the rest of this section, we prove the sufficient condition and assume that $\max_{\mathbf G\in
   \mathcal G} \{\lambda(\mathbf I-\mathbf G) \}< 1$. By Theorem
 \ref{theor:fjow}, it indicates that for each $\mathbf G \in \mathcal
 G$, $\mathbf G^{-1}\geq \mathbf 0$ exists, and thus $\mathbf p = \mathbf
 G^{-1}\mathbf n_{\mathbf G} \geq  \mathbf 0$ exists. For each receiver, define 
 \begin{equation*}
   \mathcal A_{i}^{k_i} = \left\{\mathbf p\in \mathbb R^N: \mathbf
 a_i^{k_i} \mathbf p \geq n_i^{k_i}, \mathbf p \geq \mathbf \mathbf 0 \right\}.
 \end{equation*}
Note that $\mathbf  a_i^{k_i}$ is a row vector as defined before. $\mathcal A_{i}^{k_i}$ is an intersection of half-spaces and thus is convex.
Our proof is based on Helly's theorem given below.

\begin{theorem}\textnormal{\textbf{(Helly's theorem)}\cite{nvkdp}.}
 Let $\mathcal F$ be a finite collection of convex sets in $\mathbb R^N$. The
 intersection of all the sets of $\mathcal F$ is non-empty if and only if any $N + 1$ of them has non-empty intersection.
\end{theorem}

In our case,
\begin{equation}
 \mathcal F = \left\{\mathcal A_{i}^{k_i}:
 i=1,\ldots,N, k_i\in \mathcal K_i\right \}.
\end{equation}
There are in total $K$ convex sets in $\mathcal F$. If any
$N+1$ of them have non-empty intersection, then all of them have
non-empty intersection, i.e., the SINR is feasible. The number of all combinations of such $N+1$ sets is
$\binom{K}{N+1}$. We first show the
proof for $N=2$. Then we use the mathematical induction to prove the general case. 

\begin{lemma}\label{Lemma:eodjj}
Suppose $\mathbf X=(X_{ij})$ is an $N\times N$ matrix satisfying $X_{ij}=1$
for $i=j$ and $X_{ij}\leq 0$ for $i\neq j$. Let $\mathcal S$ be a subset
of $\{1,\ldots,N\}$ and $\mathbf X'$ be the matrix by removing the
$i$-th row and $i$-th column of $\mathbf X$ for all $i\in \mathcal S$. If $\lambda(\mathbf I -
\mathbf X)<1$, then $\lambda(\mathbf I -
\mathbf X')<1$.
\end{lemma}
\begin{proof}
  Since $\lambda(\mathbf I - \mathbf X)<1$, by Theorem
  \ref{theor:fjow}, there exists a vector $\mathbf p\geq \mathbf 0$
  such that $\mathbf X\mathbf p\geq \mathbf 0$. Let $\mathbf p'\in
  \mathbb R^{N-|\mathcal S|}$ be the vector constructed by removing
  the $i$-th entry in $\mathbf p$ for all $i\in \mathcal S$. Since
  $X_{ij}\leq 0$ for $i\neq j$, it can be verified that $\mathbf
  X'\mathbf p'\geq  \mathbf 0$, which implies $\lambda(\mathbf I -
\mathbf X')<1$.
\end{proof}

\begin{lemma}\label{lemma:nncnc}
 Consider $\hat{\mathbf G}, \tilde{\mathbf G} \in \mathcal G$ such that
 $\hat{\mathbf G}$ differs from $\tilde{\mathbf G}$ only in one row. i.e., $\hat{k}_i \neq
 \tilde{k}_i$ for one $i \in \{1,\ldots, N\}$ and $\hat{k}_j = \tilde{k}_j$
 for $j\neq i$. Let $\hat{\mathbf p} = \hat{\mathbf G}^{-1}\mathbf
 n_{\hat{\mathbf G}}$ and $\tilde{\mathbf p} = \tilde{\mathbf G}^{-1}\mathbf
 n_{\tilde{\mathbf G}}$. There exists $\mathbf p\in \{\hat{\mathbf
  p},\tilde{\mathbf p}\}$ such that $\tilde{\mathbf G}\mathbf
  p\geq  \mathbf n_{\tilde{\mathbf G}}$ and $\hat{\mathbf G}\mathbf p\geq
 \mathbf n_{\hat{\mathbf G}}$.
\end{lemma}
\begin{proof}
Since $\hat{\mathbf p} = \hat{\mathbf G}^{-1}\mathbf
 n_{\hat{\mathbf G}}$ and $\tilde{\mathbf p} = \tilde{\mathbf G}^{-1}\mathbf
 n_{\tilde{\mathbf G}}$, we have $\hat{\mathbf G}\hat{\mathbf p} = \mathbf
 n_{\hat{\mathbf G}}$ and $\tilde{\mathbf G}\tilde{\mathbf p} = \mathbf
 n_{\tilde{\mathbf G}}$. If $\hat{\mathbf p} = \tilde{\mathbf p}$,
 automatically we have $\hat{\mathbf G}\tilde{\mathbf p} =\mathbf
 n_{\hat{\mathbf G}}$ and $\tilde{\mathbf G}\hat{\mathbf p}=\mathbf
 n_{\tilde{\mathbf G}}$. In the following discussion, we consider the case when $\hat{\mathbf p}\neq \tilde{\mathbf p}$.
  Without loss of generality, assume that $\hat{\mathbf G}$ and $\tilde{\mathbf G}$ differ in the first row, that is, $\hat{k}_1 \neq
 \tilde{k}_1$ and $\hat{k}_j = \tilde{k}_j$
 for $j\neq 1$. Let us partition $\hat{\mathbf G}$ into four blocks as follows.
 \begin{align*}
\hat{\mathbf G} = 
\left[
\begin{array}{c|ccc}
1 &-\mu_1\frac{g_{r_1^{\hat{k}_1},t_2}}{g_{r_1^{\hat{k}_1},t_1}} & \cdots &  -\mu_1\frac{g_{r_1^{\hat{k}_1},t_N}}{g_{r_1^{\hat{k}_1},t_1}}\\\hline
 -\mu_2\frac{g_{r_2^{\hat{k}_2},t_1}}{g_{r_2^{\hat{k}_2},t_2}}&1  & \cdots &  -\mu_2\frac{g_{r_2^{\hat{k}_2},t_N}}{g_{r_2^{\hat{k}_2},t_2}}\\
\vdots & \vdots& \ddots & \vdots \\
-\mu_N\frac{g_{r_N^{\hat{k}_N},t_1}}{g_{r_N^{\hat{k}_N},t_N}}&-\mu_N\frac{g_{r_N^{\hat{k}_N},t_2}}{g_{r_N^{\hat{k}_N},t_N}}  & \cdots &  1
\end{array}
\right]=\left[
\begin{array}{c|c}
1  &A \\ \hline
 C &D
\end{array}\right] .
 \end{align*}
Similarly, $\tilde{\mathbf G}$ is partitioned into four blocks as $\tilde{\mathbf G} =\left[
\begin{array}{c|c}
1  & B\\ \hline
 C &D
\end{array}\right]$.
Note that $\hat{\mathbf G}$ and $\tilde{\mathbf G}$ share the same
three blocks: $1, C$ and $D$. We consider  $\hat{\mathbf G}\tilde{\mathbf p} - \mathbf
 n_{\hat{\mathbf G}}$ and $\tilde{\mathbf G}\hat{\mathbf p} - \mathbf
 n_{\tilde{\mathbf G}}$. Since $\mathbf a_{i}^{\tilde{k}_i} = \mathbf
 a_{i}^{\hat{k}_i}$ and $n_{\tilde{k}_i}=n_{\hat{k}_i}$ for
 $i=2,\ldots,N$, $\mathbf a_{i}^{\hat{k}_i}\tilde{\mathbf p}
 =\mathbf a_{i}^{\tilde{k}_i}\tilde{\mathbf p}=  n_{\tilde{k}_i}=
 n_{\hat{k}_i}$ and $\mathbf a_{i}^{\tilde{k}_i}\hat{\mathbf p}
 =\mathbf a_{i}^{\hat{k}_i}\hat{\mathbf p}= n_{\hat{k}_i}=
 n_{\tilde{k}_i}$ for $i=2,\ldots,N$. If $\mathbf a_{1}^{\hat{k}_1}\tilde{\mathbf p} = n_{\hat{k}_1}$ or $\mathbf a_{1}^{\tilde{k}_1}\hat{\mathbf p} = n_{\tilde{k}_1}$,
 then $\hat{\mathbf G}\tilde{\mathbf p}= \mathbf
 n_{\hat{\mathbf G}}$ or $\tilde{\mathbf G}\hat{\mathbf p} = \mathbf
 n_{\tilde{\mathbf G}}$, which implies $\hat{\mathbf p} =
 \tilde{\mathbf p}$. Therefore when 
$\hat{\mathbf p}\neq \tilde{\mathbf p}$,
 we must have $\mathbf a_{1}^{\tilde{k}_1}\hat{\mathbf p} \neq n_{\tilde{k}_1}$ and $\mathbf
a_{1}^{\hat{k}_1}\tilde{\mathbf p}\neq n_{\hat{k}_1}$. In the
following we prove that, either $\mathbf
a_{1}^{\tilde{k}_1}\hat{\mathbf p} > n_{\tilde{k}_1}$ or $\mathbf
a_{1}^{\hat{k}_1}\tilde{\mathbf p} > n_{\hat{k}_1}$ but not both, that is, $(\mathbf
a_{1}^{\tilde{k}_1}\hat{\mathbf p} - n_{\tilde{k}_1})(\mathbf
a_{1}^{\hat{k}_1}\tilde{\mathbf p} - n_{\hat{k}_1})< 0$.

Since $\hat{\mathbf G}^{-1}>0$ exists, by Theorem \ref{theor:fjow} and Lemma \ref{Lemma:eodjj}, $D^{-1}$ exists. By block-wise inversion \cite{38924}, the inverse of $\hat{\mathbf G}$ can be written as
\begin{equation*}
\hat{\mathbf G}^{-1} = \begin{bmatrix}
a  & -aAD^{-1}\\ 
 -D^{-1}Ca & D^{-1}+D^{-1}CaAD^{-1}
\end{bmatrix},
\end{equation*}
where $a = (1-AD^{-1}C)^{-1}>0$. $\tilde{\mathbf G}^{-1}$ is in the same form by replacing $A$ with $B$ and replacing $a$ with $b=(1-BD^{-1}C)^{-1}>0$.  Denote $\mathbf n_{\hat{\mathbf G}} = \begin{bmatrix}
n_{\hat{k}_1}\\ \mathbf n'\end{bmatrix}$ and $\mathbf n_{\tilde{\mathbf G}} =\begin{bmatrix}
n_{\tilde{k}_1}\\ \mathbf n'\end{bmatrix}$ where $\mathbf n'\in\mathbb R^{N-1}$. We have 
 \begin{align*}
&\quad \mathbf a_{1}^{\tilde{k}_1}\hat{\mathbf p} - n_{\tilde{k}_1}\\
&=
\begin{bmatrix}
1  & B
\end{bmatrix}
\hat{\mathbf G}^{-1} \mathbf n_{\hat{\mathbf G}} - n_{\tilde{k}_1}\\
&=
\begin{bmatrix}
1  & B
\end{bmatrix}
\begin{bmatrix}
a  & -aAD^{-1}\\ 
 -D^{-1}Ca & D^{-1}+D^{-1}CaAD^{-1}
\end{bmatrix}
\begin{bmatrix}
n_{\hat{k}_1}\\ \mathbf n'\end{bmatrix}  - n_{\tilde{k}_1} \\
& = a n_{\hat{k}_1} - B D^{-1}Ca n_{\hat{k}_1} -aAD^{-1} \mathbf
n' + BD^{-1}\mathbf n' +  B D^{-1}CaAD^{-1}\mathbf n' -
n_{\tilde{k}_1}\\
& = a n_{\hat{k}_1} (1- B D^{-1}C) - (1 - B D^{-1}C) aAD^{-1} \mathbf n'+ BD^{-1}\mathbf n' - n_{\tilde{k}_1}\\
& =- ab^{-1}(AD^{-1} \mathbf n' -n_{\hat{k}_1}) + (BD^{-1}\mathbf n' - n_{\tilde{k}_1}).
\end{align*}
Similarly, we have
\begin{equation*}
 \mathbf a_{1}^{\hat{k}_1}\tilde{\mathbf p}- n_{\hat{k}_1}
  =- ba^{-1}(BD^{-1} \mathbf n' -n_{\tilde{k}_1}) + (AD^{-1}\mathbf n' - n_{\hat{k}_1}).
\end{equation*}
Then 
\begin{align*}
 & \quad (\mathbf
a_{1}^{\tilde{k}_1}\hat{\mathbf p}- n_{\tilde{k}_1})(\mathbf
a_{1}^{\hat{k}_1}\tilde{\mathbf p}- n_{\hat{k}_1}) \\
&= -ab^{-1}(AD^{-1}
\mathbf n' -n_{\hat{k}_1})^2 - ba^{-1}(BD^{-1} \mathbf n'
-n_{\tilde{k}_1})^2 + 2 (AD^{-1}
\mathbf n' -n_{\hat{k}_1}) (BD^{-1} \mathbf n'
-n_{\tilde{k}_1})\\
& = -\left[\sqrt{ab^{-1}}(AD^{-1}
\mathbf n' -n_{\hat{k}_1}) - \sqrt{ba^{-1}}(BD^{-1} \mathbf n'
-n_{\tilde{k}_1})\right]^2 \\
&\leq 0.
\end{align*}
Further, since $\mathbf a_{1}^{\tilde{k}_1}\hat{\mathbf p} \neq n_{\tilde{k}_1}$ and $\mathbf
a_{1}^{\hat{k}_1}\tilde{\mathbf p}\neq n_{\hat{k}_1}$, $(\mathbf
a_{1}^{\tilde{k}_1}\hat{\mathbf p}- n_{\tilde{k}_1})(\mathbf
a_{1}^{\hat{k}_1}\tilde{\mathbf p} - n_{\hat{k}_1})<0$.
In summary, there exists $\mathbf p\in \{\hat{\mathbf
  p},\tilde{\mathbf p}\}$ such that $\tilde{\mathbf G}\mathbf
  p\geq  \mathbf n_{\tilde{\mathbf G}}$ and $\hat{\mathbf G}\mathbf p\geq
 \mathbf n_{\hat{\mathbf G}}$.
\end{proof}

\subsection{Two Multicast Sessions $N=2$}
If $K=2$, i.e., $K_1=K_2=1$, it is the unicast scenario and Theorem \ref{theor:fjiwo} is true
straightforwardly. If $K_1=1,K_2=2$ or $K_1=2,K_2=1$ or $K_1=K_2=2$, then any
three subsets of $\mathcal F$ must be $\mathcal A_i^{1}, \mathcal
A_i^{2}, \mathcal A_j^{k_j}$ for $i=1$ or $2$ and $j\neq i$. Let
\begin{equation*}
  \hat{\mathbf G} = \begin{bmatrix} \mathbf a_i^1 \\
  \mathbf a_j^{k_j} \end{bmatrix} \ \text{and} \quad
\tilde{\mathbf G} = \begin{bmatrix} \mathbf a_i^2 \\
  \mathbf a_j^{k_j} \end{bmatrix}.
\end{equation*}
By Lemma \ref{lemma:nncnc}, there exists $\mathbf p$ such that $\tilde{\mathbf G}\mathbf p\geq
 \mathbf n_{\tilde{\mathbf G}}$ and $\hat{\mathbf G}\mathbf p\geq \mathbf
 n_{\hat{\mathbf G}}$, which implies $\mathbf p\in( \mathcal A_i^{1}\cap\mathcal
A_i^{2}\cap\mathcal A_j^{k_j})$. Further by Helly's theorem, the
intersection of all sets in $\mathcal F$ is non-empty.
For other values of $K_1$ and $K_2$, we divide the
$\binom{K_1+K_2}{3}$ combinations of three sets of $\mathcal F$ into two parts: 1)
two sets belong to transmitter $T_i$ and one set belongs to transmitter $T_j$ where $j\neq i$; 2) three
sets belong to the same transmitter $T_i$ for $i =1$ or
$2$. In the first case, the three sets could be $\mathcal A_i^{k_i}, \mathcal
A_i^{k_i'}, \mathcal A_j^{k_j}$. We use the same argument as before,
and conclude that the three sets have a non-empty intersection. In the
second case, the three sets could be $\mathcal A_i^{k_i}, \mathcal
A_i^{k_i'}, \mathcal A_i^{k_i''}$. It is easy to verified that $p_i =
\max\{n_i^{k_i},n_i^{k_i'},n_i^{k_i''}\}$ and $p_j = 0$ is one of
their intersection points. Overall, we prove that any three sets of
$\mathcal F$ have a non-empty intersection, and thus the intersection of all
sets is non-empty.

\subsection{Multicast Sessions with general $N$}
We use mathematical induction to
prove Theorem \ref{theor:fjiwo}. We already show that it is true when
$N=2$. Assume that the theorem holds for all numbers less than
or equal to $N-1$ and now we prove that it
also holds for $N$. If $K_i=1$ for all $i$, it is the unicast scenario and Theorem
\ref{theor:fjiwo} is true. Otherwise, we categorize the combinations of
$N+1$ sets of $\mathcal F$ into $N$ parts: 1) Receivers of $N$ transmitters are
involved: $\mathcal
A_1^{k_1}, \mathcal A_2^{k_2},\ldots,\mathcal A_N^{k_N}, \mathcal
A_i^{k_i'}$. 2) Receivers of $N-1$ transmitters are involved: $\mathcal
A_1^{k_1},\ldots,\mathcal A_{j-1}^{k_{j-1}}, \mathcal A_{j+1}^{k_{j+1}},\ldots,\mathcal A_N^{k_N}, \mathcal
A_i^{k_i'}, \mathcal A_l^{k_l'}$ where $i,l \neq j$. $\cdots$ D)
Receivers of $N-D+1$ transmitters are
involved. $\cdots$ N) Receivers of $1$ transmitter is involved.

We prove the first part.  Let
\begin{equation*}
  \hat{\mathbf G} = \begin{bmatrix} \mathbf a_1^{k_1} \\\vdots\\
\mathbf a_i^{k_i} \\ \vdots\\
  \mathbf a_N^{N_j} \end{bmatrix} \ \text{and} \quad
\tilde{\mathbf G} =\begin{bmatrix} \mathbf a_1^{k_1} \\\vdots\\
\mathbf a_i^{k_i'} \\ \vdots\\
  \mathbf a_N^{N_j} \end{bmatrix}.
\end{equation*}
By Lemma \ref{lemma:nncnc}, there exists $\mathbf p$ such that $\tilde{\mathbf G}\mathbf p\geq
 \mathbf n_{\tilde{G}}$ and $\hat{\mathbf G}\mathbf p\geq \mathbf
 n_{\hat{G}}$, which implies $\mathbf p\in (\cap_{j=1}^N\mathcal
A_j^{k_j}\cap\mathcal A_i^{k_i'})$.

We prove the D) part for $D=2,\ldots,N$. Suppose the $D-1$
transmitters whose receivers are not involved in the $N+1$
sets, are $d_1,d_2,\ldots,d_{D-1} \in \{1,\cdots,N\}$. We simply
let $p_{d_1}=p_{d_2}=\cdots=p_{d_{D-1}}=0$. The resulting system is equivalent
to having $N-D+1$ multicast sections characterized by matrix
$\mathbf A'$, which is constructed by removing the rows in $\mathbf A$ that
corresponds to the receivers of transmitter $d$ and the $d$-th column
of $\mathbf A$, for $d=d_1,\ldots,d_{D-1}$. Define $\mathcal G'
\subset\mathbb R^{(N-D+1)\times (N-D+1)}$ for $\mathbf A'$. For any
$\mathbf G'\in \mathcal G'$, we can find a $\mathbf G\in \mathcal G$
such that, $\mathbf G'$ is constructed by removing the $d$-th row and
$d$-th column of $\mathbf G$ for all $d=d_1,\ldots,d_{D-1}$. Since
$\lambda(\mathbf I-\mathbf G)<1$ for all $\mathbf G\in \mathcal G$, by Lemma
\ref{Lemma:eodjj}, $\lambda(\mathbf I-\mathbf G')<1$, and therefore $\max_{\mathbf G'\in
   \mathcal G'} \{\lambda(\mathbf I- \mathbf G') \}< 1$. By the inductive hypothesis, we can apply Theorem \ref{theor:fjiwo} when $N-D+1<N$,
 and thus there exists $\mathbf p'\geq
\mathbf 0$ such that $\mathbf A'\mathbf p'\geq \mathbf n'$, where
$\mathbf n'$ is obtained by removing the entries that
correspond to the receivers of transmitter $T_d$ for $d=d_1,\ldots,d_{D-1}$. By
inserting $0$ back into $\mathbf p'$ at the position of transmitter $T_d$
for all $d=d_1,\ldots,d_{D-1}$, we get a power $\mathbf p\geq \mathbf
0$ which is in the $N+1$ subsets.

Overall, we have proved that any $N+1$ subsets of $\mathcal F$ has a non-empty
intersection. By Helly's theorem, all subsets in $\mathcal F$ has an
intersection. This completes the proof
of Theorem \ref{theor:fjiwo}.

\section{Feasible SINR region and algorithm}\label{section:algorithm}
In this section, we characterize the feasible SINR region of a
wireless multicast system by
analytically obtaining its boundary points. By Proposition
\ref{pro:aofis}, we know that the feasible SINR region is downward
comprehensive. That is, if $\boldsymbol\mu$ is feasible, then any $\boldsymbol\mu'$ satisfying
$\mathbf 0\leq\boldsymbol\mu'\leq\boldsymbol\mu$ is also
feasible. Therefore, finding the boundary points is enough to figure
out the feasible SINR region. Our approach is to find the
farthest point from the origin in a given direction. In mathematics,
the problem is formulated as
\begin{align*}
\sup_{\mathbf p} \quad & \beta\\
s.t. \quad & \mathbf A(\beta\boldsymbol\mu) \mathbf p\geq \mathbf
             n(\beta\boldsymbol\mu)\\
& \mathbf p\geq \mathbf 0,     
\end{align*}
where $\boldsymbol\mu$ is a given direction. By Theorem
\ref{theor:fjiwo}, there is a feasible solution to the above problem
if and only if 
\begin{equation*}
\max_{\mathbf G\in
   \mathcal G(\beta\boldsymbol\mu)} \{\lambda(\mathbf I-\mathbf G) \}
 = \beta \cdot \max_{\mathbf G\in
   \mathcal G(\boldsymbol\mu)} \{\lambda(\mathbf I-\mathbf G) \} < 1.
\end{equation*}
That is 
\begin{equation*}
\beta<\frac{1}{\max_{\mathbf G\in
   \mathcal G(\boldsymbol\mu)} \{\lambda(\mathbf I-\mathbf G) \}}.
\end{equation*}
Therefore, the optimal value is 
\begin{equation*}
\beta^*(\boldsymbol\mu) = \frac{1}{\max_{\mathbf G\in
   \mathcal G(\boldsymbol\mu)} \{\lambda(\mathbf I-\mathbf G) \}}.
\end{equation*}
$\beta^*(\boldsymbol\mu) \boldsymbol\mu$ is a boundary
point of the SINR region.
The open line segment defined by
$\{\alpha\boldsymbol\mu:0<\alpha<\beta^*(\boldsymbol\mu) \}$ is in the
feasible SINR region $\Upsilon$, but $\alpha\boldsymbol\mu$ is not in the
feasible region if $\alpha>\beta^*(\boldsymbol\mu)$.

We note that the size of $\mathcal G$ is $\prod_{i=1}^N K_i$, which grows exponentially with $N$. It is not an efficient method to calculate
the Perron-Frobenius eigenvalue of all the embedded unicast systems
and find out the maximum one. Next, we propose
an iterative algorithm to compute $\beta^*(\boldsymbol\mu)$.
For $i = 1,2,\ldots,N$, let $\mathbf e_i$ denote the $N$-dimensional column vector such that the
$i$-th component of $\mathbf e_i$ is $1$ while the others are $0$. The
algorithm is described in Algorithm \ref{iter}.

\begin{algorithm}
\caption{Iterative algorithm}\label{iter}
\begin{algorithmic}[1]
\State Choose $\mathbf p^{(0)}\in \mathbb R^N >\mathbf 0$ and $k\gets
0$
\Repeat
\For {$i=1$ to $N$}
\State $ y_i^{(k)} \gets \max_{k_i\in \mathcal K_i}\left\{ \big(\mathbf
e_i^T-\mathbf a_i^{k_i}\big)\mathbf p^{(k) }\right\} $
\EndFor
\State $\beta^{(k)}\gets \min_{i=1}^N\left\{\frac {p_i^{(k)}}{y_i^{(k)}}\right\} $
\State $\mathbf p^{(k+1)} \gets \frac{\mathbf y^{(k)}}{||\mathbf
  y^{(k)}||}$
\State $k\gets k+1$
\Until{convergence}
 \State \textbf{return} $\beta^{(k)}$
\end{algorithmic}
\end{algorithm}

For receiver $R_i^{k_i}$, $(\mathbf e_i^T-\mathbf a_i^{k_i})\mathbf
p^{(k) }$ is the sum of the interference power. The power of transmitter
$T_i$ is updated by the maximum interference power experienced by the
receivers in its multicast session. This idea is similar
to the distributed power control algorithm for unicast systems
\cite{577019} to solve the power balancing problem. Recall
that in \cite{577019}, given a normalized interference link gain
matrix $\mathbf I-\mathbf G$, the
algorithm works as $\mathbf p^{(k+1)} = \frac{(\mathbf I-\mathbf G)\mathbf
  p^{(k)}}{||(\mathbf I-\mathbf G)\mathbf p^{(k)}||}$, where $k$ is the iteration
index. It is well known that when $\mathbf I-\mathbf G$ is primitive (to be
defined later), $||(\mathbf I-\mathbf G)\mathbf p^{(k)}||$ converges to the Perron-Frobenius eigenvalue of
$\mathbf I-\mathbf G$, and $\mathbf p^{(k)}$ converges to the corresponding
eigenvector. In our proposed algorithm, we are dealing with multicast systems.
For notation simplicity, define
\begin{equation}\label{equ:zz}
\mathcal Z(\boldsymbol\mu) = \{ \mathbf Z = \mathbf I-\mathbf G:
\mathbf G\in \mathcal G(\boldsymbol\mu)\}.
\end{equation}
$\mathcal Z$ includes the normalized interference link gain
matrices of all the embedded unicast systems and $\mathbf Z\geq
\mathbf 0$ for all $\mathbf Z\in \mathcal Z$. Given any $\mathbf p>0$, due to the structure of $ \mathcal G(\boldsymbol\mu)$, there always exists $\hat{\mathbf Z}\in \mathcal Z$ such that $\hat{\mathbf Z} \mathbf p\geq \mathbf Z \mathbf p$ for all $\mathbf Z\in \mathcal Z$. Our algorithm works as 
$\mathbf p^{(k+1)} = \frac{\mathbf Z^{(k)}\mathbf p^{(k)}}{||\mathbf
Z^{(k)}\mathbf p^{(k)}||}$, where $\mathbf Z^{(k)}\in \mathcal Z$ is
chosen such that $\mathbf Z^{(k)} \mathbf p^{(k)} \geq \mathbf Z
\mathbf p^{(k)}$ for all $\mathbf Z\in \mathcal Z$. In the rest of this section, we show the convergence of the algorithm.

\begin{lemma}\label{lemma:decreasing}
The sequence $\{\beta^{(k)}\}$ generated by Algorithm \ref{iter} is
monotonically increasing and bounded above by $\frac{1}{\max_{\mathbf G\in
   \mathcal G(\boldsymbol\mu)} \{\lambda(\mathbf I-\mathbf G) \}}=\frac{1}{\max_{\mathbf Z\in
   \mathcal Z(\boldsymbol\mu)} \{\lambda(\mathbf Z) \}}$, and thus is convergent. 
\end{lemma}
\begin{proof}
By Algorithm \ref{iter}, we have $\mathbf y^{(k)}\geq \mathbf Z\mathbf
p^{(k)}$ for all $\mathbf Z\in \mathcal Z$, and $\mathbf p^{(k)}\geq
\beta^{(k)}\mathbf y^{(k)}$.
Then 
\begin{equation*}
\mathbf Z\mathbf p^{(k+1)} = \mathbf Z \frac{\mathbf
  y^{(k)}}{||\mathbf y^{(k)}||} \leq  \mathbf Z \frac{ \mathbf
  p^{(k)}}{\beta^{(k)}||\mathbf y^{(k)}||} \leq \frac{\mathbf
  y^{(k)}}{\beta^{(k)}||\mathbf y^{(k)}||} = \frac{\mathbf p^{(k+1)}}{\beta^{(k)}}.
\end{equation*}
Since the above inequality holds for all $\mathbf Z\in \mathcal Z$, we have
\begin{equation*}
\beta^{(k)}\leq \min_{\mathbf Z\in \mathcal Z}\left\{\min_{i=1}^N\Big\{\frac{p_i^{(k+1)}}{[\mathbf Z\mathbf p^{(k+1)^T}]_i}\Big\}\right\} = \min_{i=1}^N\left\{\frac{p_i^{(k+1)}}{ y_i^{(k+1)}}\right\} =\beta^{(k+1)}.
\end{equation*}
That is, $\{\beta^{(k)}\}$ is monotonically increasing. On the other
hand, since $\mathbf p^{(k)}\geq \beta^{(k)}\mathbf y^{(k)}\geq
\beta^{(k)}\mathbf Z\mathbf p^{(k)} $ for all $\mathbf Z\in \mathcal
Z$, that is $(\mathbf I-\beta^{(k)}\mathbf Z)\mathbf p^{(k)}\geq
\mathbf 0$, we have $\lambda(\beta^{(k)}\mathbf Z)\leq 1$ by Theorem \ref{theor:fjow}. Therefore $\beta^{(k)}\leq\frac{1}{\max_{\mathbf Z\in
   \mathcal Z(\boldsymbol\mu)} \{\lambda(\mathbf Z) \}}$, and thus $\{\beta^{(k)}\}$ is convergent.
\end{proof}

Denote $\lim_{k\rightarrow \infty} \beta^{(k)} = \beta^*$. Before we proceed to show that $\beta^* = \frac{1}{\max_{\mathbf Z\in
   \mathcal Z(\boldsymbol\mu)} \{\lambda(\mathbf Z) \}}$, we introduce
the concept of \emph{primitive} matrix and \emph{primitive}  set.

\begin{definition}{\textnormal{\cite{92048}}}
A square nonnegative matrix $\mathbf X$ is called \emph{primitive} if there exists a positive integer $n$ such that $\mathbf X^n>0$.
\end{definition}

The class of primitive matrices is a subclass of irreducible
matrices. If $\mathbf X$ is primitive, then its Perron-Frobenius
eigenvalue is strictly greater than all other eigenvalues in absolute
value. The primitive condition guarantees the convergence of the
aforementioned distributed power control algorithm for unicast
systems. In our multicast case, we need to use the concept of
primitive set, which replaces a single matrix and powers of that matrix with a set of matrices and inhomogeneous products of matrices from the set. 

\begin{definition}{\textnormal{\cite{Prim}}}
Let $\mathcal Z$ be a set of $N\times N$ nonnegative matrices. For a
positive integer $n$, let $\Theta(n)$ be an arbitrary product of $n$ matrices from $\mathcal Z$, with any
ordering and with repetitions permitted. Define $\mathcal Z$ to be a \emph{primitive} set if there is a
positive integer $n$ such that every $\Theta(n)$ is positive. 
\end{definition}

It can be seen that a necessary condition for $\mathcal Z$ to be primitive is
that $\mathbf Z$ is primitive for all $\mathbf Z\in \mathcal Z$. One of the sufficient conditions for $\mathcal Z$ to be primitive is
that for any $\mathbf Z\in \mathcal Z$, in each row and each column of
$\mathbf Z$, there are more than half of the entries that are positive
\cite{Prim}. Interested readers can refer to \cite{Prim} for more
information of the primitive set. It needs to be mentioned that when
the system is composed of two multicast sessions, $\mathbf Z \in \mathcal Z$ are always
non-primitive, and therefore $\mathcal Z$ cannot be
primitive. However, for this case, Corollary \ref{corol} already gives
an explicit and simple solution to the feasible SINR region. Algorithm
\ref{iter} works for systems with more than two multicast sessions and with $\mathcal Z$ being primitive.

\begin{theorem}\label{theorem:loveh}
If the matrix set $\mathcal Z(\boldsymbol\mu)$ defined in \eqref{equ:zz} is primitive, then $\beta^{(k)}$ converges to $\beta^*=\frac{1}{\max_{\mathbf Z\in
   \mathcal Z(\boldsymbol\mu)} \{\lambda(\mathbf Z) \}}$ for
an arbitrary initial value $\mathbf p^{(0)}>0$. Moreover, $\mathbf
p^{(k)}$ converges to a power vector $\mathbf p^*$ such that $\lim_{\alpha\rightarrow
  \infty} \Gamma(\alpha \mathbf p^*)$ achieves the boundary point $\beta^*\boldsymbol\mu$.
\end{theorem}

The proof is provided in Appendix \ref{section:naofffa}.
Fig. \ref{fig:con_alg} illustrates the typical behavior of Iterative
Algorithm \ref{iter}. In this example, there are four multicast
sessions and each has three receivers. The link gains are randomly drawn from a uniform distribution on $[0, 1)$.
 As we can see, $\beta^{(k)}$ converges within a
small number of iterations.  By using Algorithm \ref{iter}, we can efficiently check the feasibility of
an SINR vector $\boldsymbol\mu$ by checking the value of
$\beta^*(\boldsymbol\mu)$. If $\beta^*(\boldsymbol\mu)<1$,
$\boldsymbol\mu$ is infeasible, and vice versa. Besides, Algorithm
\ref{iter} can be used to find the optimal solution of the classic power balancing
problem for multicast systems, which is in the following form
\begin{align*}
\sup_{\mathbf p} \min_{i=1}^N &\quad\gamma_i(\mathbf p)\\
s.t. &\quad \mathbf p \geq\mathbf 0,
\end{align*}
and the solution is $\beta^*(\mathbf 1)$.

\begin{figure}
 \centering
 \begin{tikzpicture}
   \begin{axis}[
     xlabel = iteration step $k$,
     ylabel=$\beta^{(k)}$,
    width=300pt,
    height=230pt,
     xmin=1,xmax=15,
     ymin=0.6,ymax=2,
line width=1.5pt,
grid=major]
\addplot[color=orange,mark=o] table[x=Y, y=X] {convergence.txt};
   \end{axis}
 \end{tikzpicture}
 \caption{Convergence of Algorithm \ref{iter}. In this example, there
   are four multicast sessions and each has three receivers. The link
   gains are randomly drawn from a uniform distribution on $[0,1)$.}
 \label{fig:con_alg}
\end{figure}
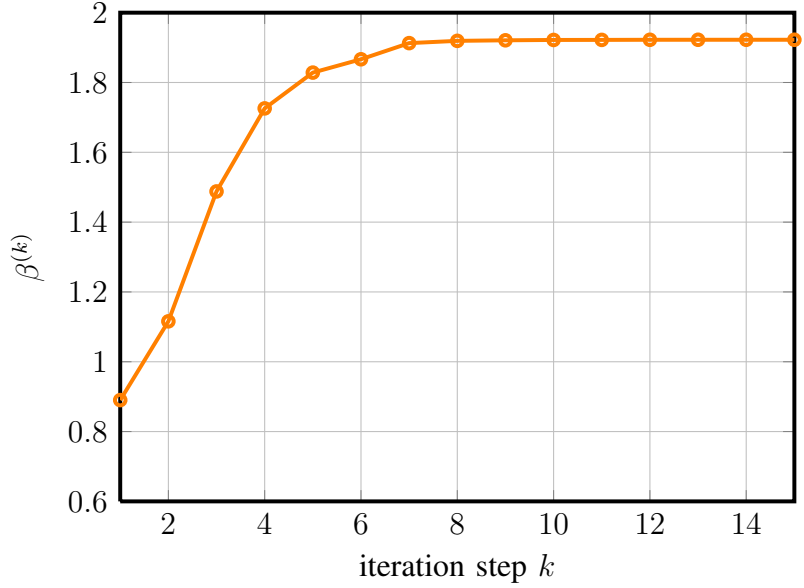

\section{Geometric Properties of the Feasible SINR
  region}\label{section:geometric}

In this section, we discuss the geometric properties of the feasible
SINR region. Let $D(\boldsymbol\mu)$ denote the diagonal matrix
constructed by 
\begin{equation*}
D(\boldsymbol\mu)=\begin{bmatrix}
 \mu_1& 0&\cdots  & 0 \\
0 & \mu_2&\cdots  & 0 \\
\vdots  & \vdots  & \ddots & \vdots  \\
0 & \cdots&0 & \mu_N
 \end{bmatrix}.
\end{equation*}
By Theorem \ref{theor:fjiwo}, the feasible SINR region is
equivalent to
\begin{align*}
\Upsilon &=\left \{\boldsymbol\mu\in \mathbb R^N:  \max_{\mathbf Z\in
   \mathcal Z(\boldsymbol\mu)} \{\lambda(\mathbf Z) \}< 1\right\}\\
&= \left\{\boldsymbol\mu\in \mathbb R^N:  \max_{\mathbf Z\in
   \mathcal Z(\mathbf 1)} \Big\{\lambda \big(D(\boldsymbol\mu)\mathbf Z\big) \Big\}< 1\right\}\\
& = \bigcap_{\mathbf Z\in
   \mathcal Z(\mathbf 1)} \Big\{\boldsymbol\mu\in \mathbb R^N:  \lambda \big(D(\boldsymbol\mu)\mathbf Z\big) < 1\Big\}.
\end{align*} 
That is, the feasible SINR region of a multicast system is the
intersection of the feasible SINR regions of all its embedded unicast systems. Let $\Upsilon^c = \mathbb R_{+}^N \setminus \Upsilon$ denote the complement of $\Upsilon$ in $\mathbb R_{+}^N$, i.e., the infeasible SINR region. 
Next, we investigate the convexity of $\Upsilon^c$ and the log-convexity of $\Upsilon$. 

\subsection{Convexity of $\Upsilon^c$}
For unicast systems, it has been proved in \cite{5730591} that the
infeasible SINR regions of a general two user system and a general
three user system are convex. It is also shown in \cite{4012511} that
the convexity of the infeasible SINR region does not hold for a
general four user system. For multicast systems, we have the following observation.

\begin{theorem}
 The infeasible SINR region of a general system consisting of two multicast sessions is convex. The convexity property does not hold for a general system consisting of more than two multicast sessions.  
\end{theorem}
When there are two multicast sessions, by Corollary \ref{corol}, the feasible SINR region is 
\begin{equation*}
\Upsilon  =\left\{[\mu_1,\mu_2]\in \mathbb R_+^2: \mu_1\mu_2 <
\frac{g_{r_1^{k_1^*},t_1}}{g_{r_1^{k_1^*},t_2}}\cdot 
\frac{g_{r_2^{k_2^*},t_2}}{g_{r_2^{k_2^*},t_1}}\right\}.
\end{equation*}
It is ready to verify that $\Upsilon^c$ is convex.

When there are more than two multicast sessions, $\Upsilon^c$ is the
union of the infeasible SINR regions of all the embedded unicast
systems and is in general non-convex. Fig. \ref{fig:safe} illustrates
the $\Upsilon^c$ for a system consisting of three multicast sessions,
where the link gain matrix is given by
\begin{equation*}
\bordermatrix[{[]}]{
& T_1& T_2 & T_3\cr
R_1^1& 1 & 0.5 & 0.1\cr
R_1^2&1& 0.1 & 0.5\cr
R_2^1& 0.5 & 1& 0.1 \cr
R_2^2& 0.1&1 &0.5\cr
R_3^1& 0.5 & 0.1&1\cr
R_3^2& 0.1 & 0.5 &1}.
\end{equation*}
It can be seen that its $\Upsilon^c$ is non-convex.

\subsection{Log-convexity of $\Upsilon$}
We first introduce the notion of log-convexity. Let $\log(
\boldsymbol\mu) = [\log\mu_1,\log\mu_2,\ldots,\log\mu_N]$ and
$\log(\Upsilon) = \{\log(
\boldsymbol\mu): \boldsymbol\mu\in \Upsilon\}$. 
We say a set
$\Upsilon$ is log-convex if $\log(\Upsilon)$ is convex. Since
$\log(\cdot): \Upsilon\rightarrow \log(\Upsilon)$ is a bijective mapping, we have 
\begin{equation*}
\log(\Upsilon) = \bigcap_{\mathbf Z\in
   \mathcal Z(\mathbf 1)} \Big\{\log(\boldsymbol\mu)\in \mathbb R^N:  \lambda \big(D(\boldsymbol\mu)\mathbf Z\big) < 1\Big\}.
\end{equation*}
It has been proved in \cite{1010870} that the feasible SINR region of
a unicast system is log-convex. So $\log(\Upsilon)$, the intersection
of the SINR regions of all its embedded unicast, is also
log-convex. We conclude this by the following theorem.

\begin{theorem}
The feasible SINR region of a multicast system is log-convex. In other words, the feasible SINR, expressed in decibels, is a convex set.
\end{theorem}

\begin{figure}
 \centering
     \includegraphics[width=0.6\textwidth]{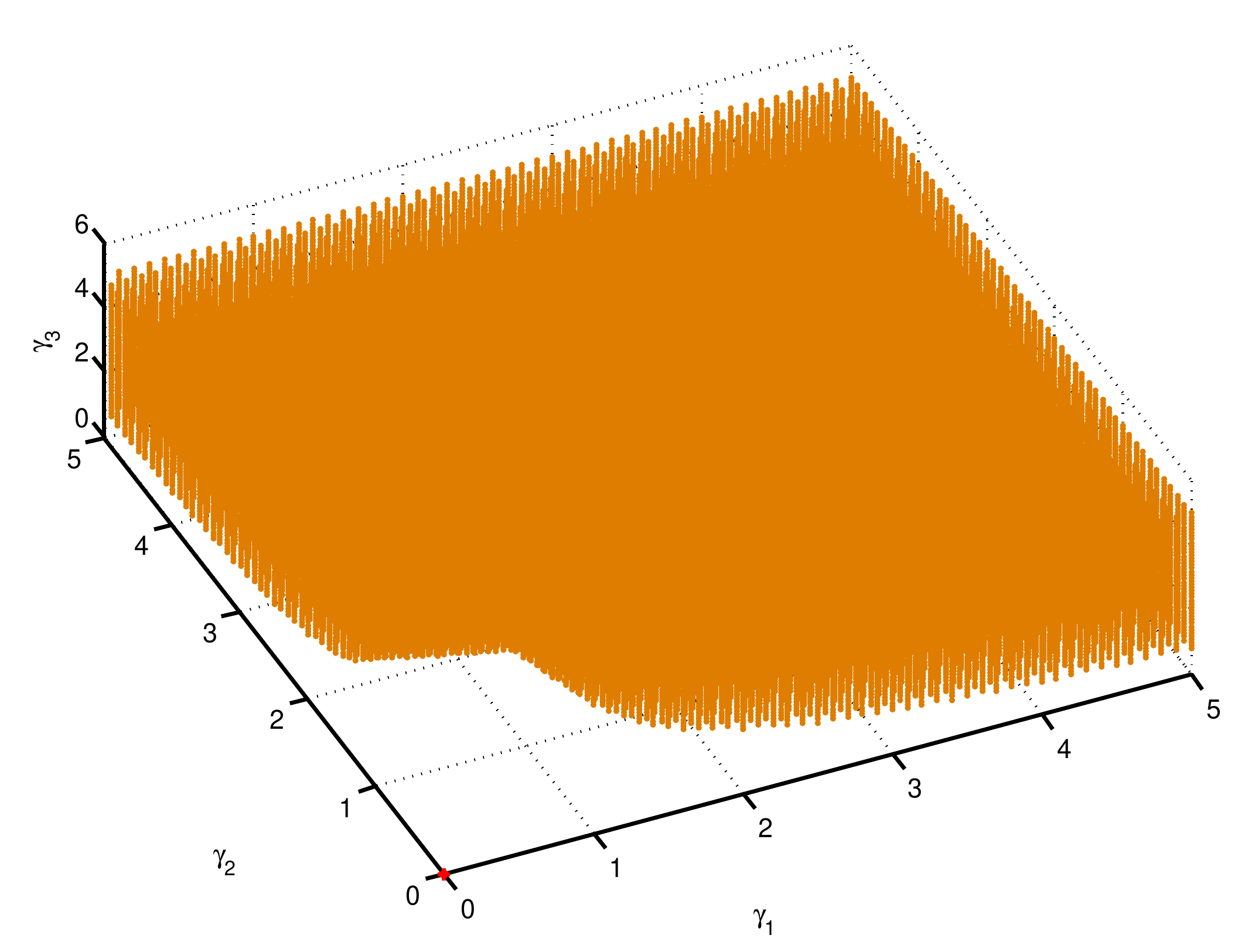}
 \caption{Infeasible SINR region of a three-multicast-session system.}
\label{fig:safe}
\end{figure}

\section{Feasibility of SINR with Power Constraints}\label{section:constraint}
So far, we have discussed the feasibility of SINR for a multicast system in the case of unlimited power. 
In this section, we consider that besides $\mathbf p\geq \mathbf 0$,
the power vector are also subject to the linear constraints
\begin{equation*}
\sum_{i\in \Omega_m}p_i \leq \bar p_{\Omega_m}, m=1,\ldots,M,
\end{equation*}
where $\Omega_m\subseteq \{1,\ldots,N\}$ and $M$ is the number of constraints. When $\Omega_m=\{1,\ldots,N\}$, it is a constraint on the total power. When $\Omega_m=\{i\}$, it is a constraint on the individual power of transmitter $T_i$. Define the power set by
\begin{equation*}
\mathcal P = \{\mathbf p\geq \mathbf 0 \ \text{and} \ \sum_{i\in \Omega_m}p_i \leq \bar p_{\Omega_m}, m=1,\ldots,M\}.
\end{equation*}
Now the feasibility of SINR vector $\boldsymbol\mu$ is decided by
whether there exists  $\mathbf p\in \mathcal P$ such that $\mathbf
A(\boldsymbol\mu)\mathbf p=\mathbf n(\boldsymbol\mu)$. Note that the
power vectors in $\mathcal P$ are downward comprehensive. That is, if
$\mathbf p'\in \mathcal P$, then $\mathbf p\in \mathcal P$ if $\mathbf
0\leq \mathbf p\leq \mathbf p' $. Hence using the same argument as in
Proposition \ref{pro:aofis}, we know that $\boldsymbol\mu$ is feasible
if and only if there exits $\mathbf p\in \mathcal P$ such that $\mathbf
A(\boldsymbol\mu)\mathbf p\geq \mathbf n(\boldsymbol\mu)$. Our results
generalize the feasibility condition derived in \cite{5466510} for a
unicast system to a multicast system. 

\begin{definition}
For a matrix $\mathbf X\in\mathbb R^{K\times N}$, a vector $\mathbf
y\in\mathbb R^{K}$ and a set $\Omega\subseteq \{1,\ldots,N\}$,
$\psi(\mathbf X, \mathbf y, \Omega)$ is the operation to add $\mathbf
y$ to the $j$-th column of $\mathbf X$, for all $j\in \Omega$. That is, $\mathbf Z = \psi(\mathbf X, \mathbf y, \Omega)$, where $Z_{ij} = X_{ij} +  y_i$ for all $i\in \{1,\ldots,K\}$ and $j\in \Omega$, and $Z_{ij} = X_{ij}$ for the else.
\end{definition}

\begin{theorem}\label{theor:jiayou}
 Consider a multicast network setting $\mathbf A(\boldsymbol\mu)$ and
 assume the matrices $\mathbf I-\mathbf
 G$ for $\mathbf G\in \mathcal G(\boldsymbol\mu)$ are all
 irreducible. There exists a power vector $\mathbf p\in \mathcal P$ such that $\mathbf A(\boldsymbol\mu)
 \mathbf p\geq \mathbf n(\boldsymbol\mu)$ if and only if 
\begin{equation*}
\max_{\mathbf G\in
   \mathcal G(\boldsymbol\mu)} \max_{m\in
   \{1,\ldots,M\}}\left\{\lambda \Big(\psi \big(\mathbf I-\mathbf G,
   \frac{\mathbf n_{\mathbf G}}{\bar p_{\Omega_m}},\Omega_m \big)\Big)
 \right\} \leq 1.
\end{equation*}
\end{theorem}
\begin{proof}
It is already known from \cite{5466510} that, for a unicast system
$\mathbf G$, there exists $\mathbf p\in \mathcal P$ such that $\mathbf
G\mathbf p\geq \mathbf n_{\mathbf G}$ if and only if $\max_{m\in
  \{1,\ldots,M\}}\{\lambda (\psi (\mathbf I-\mathbf G, \frac{\mathbf
  n_{\mathbf G}}{\bar p_{\Omega_m}},\Omega_m )) \}\leq 1$. We first prove
the necessary condition. Suppose there exists $\mathbf p\in \mathcal
P$ such that $\mathbf A(\boldsymbol\mu)\mathbf p\geq \mathbf
n(\boldsymbol\mu)$. Then for any $\mathbf G\in\mathcal
G(\boldsymbol\mu)$, $\mathbf G\mathbf p\geq \mathbf n_{\mathbf G}$, which implies $\max_{m\in \{1,\ldots,M\}}\{\lambda (\psi (\mathbf I-\mathbf G, \frac{\mathbf n_{\mathbf G}}{\bar p_{\Omega_m}},\Omega_m )) \}\leq1$. Regarding all $\mathbf G\in \mathcal G(\boldsymbol\mu)$, we have $\max_{\mathbf G\in
   \mathcal G(\boldsymbol\mu)} \max_{m\in \{1,\ldots,M\}}\{\lambda (\psi (\mathbf I-\mathbf G, \frac{\mathbf n_{\mathbf G}}{\bar p_{\Omega_m}},\Omega_m)) \}leq 1$.

Next we prove the sufficient condition. For any $\mathbf G\in \mathcal
G(\boldsymbol\mu)$, since $\mathbf 0\leq\mathbf I-\mathbf G <\psi (\mathbf I-\mathbf
G, \frac{\mathbf n_{\mathbf G}}{\bar p_{\Omega_m}},\Omega_m )$ for all
$m$, by the Perron-Frobenius Theorem for irreducible matrices
\cite{92048}, $\lambda(\mathbf I-\mathbf G)<\lambda (\psi (\mathbf
I-\mathbf G, \frac{\mathbf n_{\mathbf G}}{\bar p_{\Omega_m}},\Omega_m
))\leq1$. This implies that $\mathbf G^{-1}$ exists and $\mathbf p =\mathbf
G^{-1}\mathbf n_{\mathbf G} \in \mathcal P$. The rest of the proof
follows the same argument as the proof of Theorem \ref{theor:fjiwo}.
\end{proof}

Note that 
\begin{equation*}
\max_{\mathbf G\in
   \mathcal G(\boldsymbol\mu)} \max_{m\in
   \{1,\ldots,M\}}\left\{\lambda \Big(\psi \big(\mathbf I-\mathbf G,
   \frac{\mathbf n_{\mathbf G}}{\bar p_{\Omega_m}},\Omega_m \big)\Big)
 \right\} =  \max_{m\in
   \{1,\ldots,M\}}\max_{\mathbf G\in
   \mathcal G(\boldsymbol\mu)}\left\{\lambda \Big(\psi \big(\mathbf I-\mathbf G,
   \frac{\mathbf n_{\mathbf G}}{\bar p_{\Omega_m}},\Omega_m
   \big)\Big)
 \right\}.
\end{equation*}
Similar to \eqref{equ:zz}, for each of the $M$ linear constraints, define 
\begin{equation*}
\mathcal Z_{\Omega_m}(\boldsymbol\mu) = \left\{\psi \big(\mathbf I-\mathbf G,
   \frac{\mathbf n_{\mathbf G}}{\bar p_{\Omega_m}},\Omega_m \big):
   \mathbf G\in \mathcal G(\boldsymbol\mu)\right\}.
\end{equation*}
By using Algorithm \ref{iter} with $\mathcal
Z_{\Omega_m}(\boldsymbol\mu)$, we can find a supremum
$\beta_{\Omega_m}^*(\boldsymbol\mu)$. The farthest
point of the SINR region in direction $\boldsymbol\mu$ is then
$\min_{m=1}^M\{\beta_{\Omega_m}^*(\boldsymbol\mu)\}\boldsymbol\mu$. By this approach,
the feasible SINR region is characterized. On the other hand, if
$\min_{m=1}^M\{\beta_{\Omega_m}^*(\boldsymbol\mu)\}\geq 1$, $\boldsymbol\mu$
is feasible. Fig. \ref{fig:con_SINR} plots the feasible SINR region
of the network example in Fig. \ref{fig:example}, with a power constraint on
the total power. In this example, the link gain matrix is 
\begin{equation*}
\bordermatrix[{[]}]{
& T_1& T_2 \cr
R_1^1& 0.5326   & 0.6801\cr
R_1^2&0.5539   & 0.3672\cr
R_2^1& 0.2393  &  0.8669 \cr
R_2^2& 0.5789   & 0.4068},
\end{equation*}
and the power constraint is $p_1+p_2 \leq 2$.
The four dashed lines are the boundary of the
feasible SINR regions of four embedded unicast systems and the solid
line is the boundary of the multicast system. It can be seen that
under power contraint, the infeasible SINR region is not necessary to
be convex even for a multicast system with two multicast sessions.

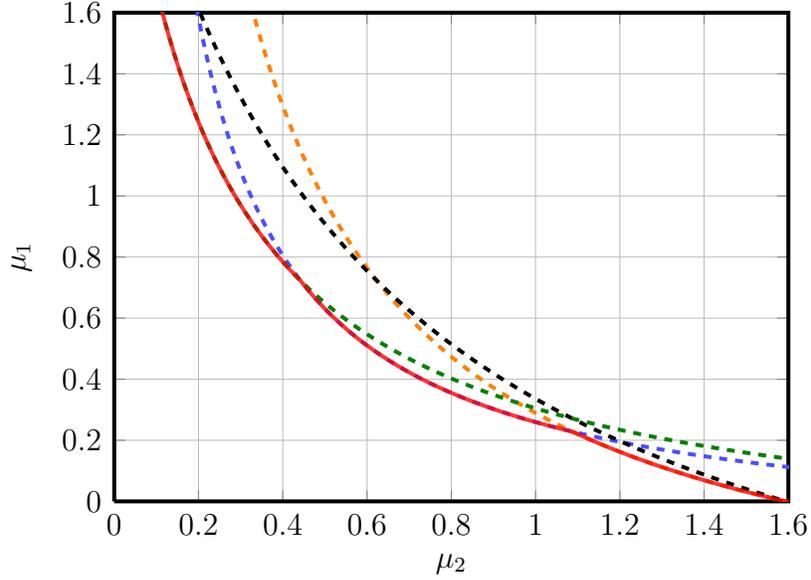
\begin{figure}
 \centering
 \begin{tikzpicture}
   \begin{axis}[
     xlabel = $\mu_2$,
     ylabel=$\mu_1$,
    width=300pt,
    height=230pt,
     xmin=0,xmax=1.6,
     ymin=0,ymax=1.6,
line width=1.5pt,
grid=major]
\addplot[color=orange,dashed] table[x=Mu1, y=Mu2_a1] {SINR_intersection2.txt};
\addplot[color=blue!70,dashed] table[x=Mu1, y=Mu2_a2] {SINR_intersection2.txt};
\addplot[color=black,dashed] table[x=Mu1, y=Mu2_a3] {SINR_intersection2.txt};
\addplot[color=green!50!black,dashed] table[x=Mu1,y=Mu2_a4]{SINR_intersection2.txt};
\addplot[color=red,opacity=.8] table[x=Mu1,y=Min]{SINR_intersection2.txt};
   \end{axis}
 \end{tikzpicture}
 \caption{Feasible SINR region for a multcast system with two
   multicast sessions, under a total power constraint. The
   dashed lines correspond to the four embedded unicast systems and the solid
line corresponds to the multicast system.}
 \label{fig:con_SINR}
\end{figure}

In the end of this section, we introduce an application of our
multicast model to a time varying unicast system. Consider a unicast
system consisting of $N$ transmitter-receiver pairs, where the channel
gains among them vary with time due to the mobility of the
receivers. Let $\mathbf h_i(t)$ for $i=1,\ldots,N$ denote the link gain vector from $N$
transmitters to the $i$-th receiver at time $t$. As argued in
\cite{5466510}, $\mathbf h_i(t)$ can be modeled with discrete states,
that is, $\mathbf h_i(t)$ is randomly selected from a finte set
$\{\mathbf h_i^1,\mathbf h_i^2,\ldots,\mathbf h_i^{K_i}\}$ for all
$i$. An SINR vector $\boldsymbol\mu$ is said to be zero-outage if there exists a power
such that no matter what the link gain realization is, the SINR is
achievable all the time. Such a zero-outage SINR problem can be mapped to a feasible
SINR problem of a multicast system. The idea is to let one receiver $R_i$
pretend to be $K_i$ receivers, i.e., $R_i^{1},\ldots,R_i^{K_i}$, and
$R_i^{k_i}$ only experiences the link gain $\mathbf h_i^{k_i}$. The
feasible SINR region of this artificial multicast system is exactly
the zero-outage SINR region of the original time-varying unicast
sytem. Theorem \ref{theor:fjiwo} and Theorem \ref{theor:jiayou} can be
applied. It needs to be mentioned that the scenario considered here is
different from that in \cite{5466510}. In \cite{5466510}, the power
can change with the channel states and is subject to an average power
constraint. In our model, the power is universal for all possible
channel states.

\section{Conclusion}\label{section:conclusion}
In this paper, we characterize the feasibility condition of an SINR
vector for a multicast system, which generalizes the Perron-Frobenius
Theorem for a unicast system. We also propose an iterative algorithm
which can efficiently check the condition and compute the boundary points of the feasible SINR
region. According to the earlier mentioned
Gaussian interference assumption, by describing the feasible SINR
region, we directly obtain the feasible rate region by applying the Shannon Capacity formula for AWGN
channels that maps the SINR to the rate.

\appendices
\section{Proof of Theorem \ref{theorem:loveh} }\label{section:naofffa}

\begin{proof}
Let $\mathbf Z^{(k)}$ denote one of the matrices at the $k$-th
iteration such that $\mathbf Z^{(k)}\mathbf p^{(k)} \geq \mathbf
Z\mathbf p^{(k)}$ for all $\mathbf Z\in \mathcal Z$. From the
construction of the algorithm, we have
\begin{equation}\label{equ:pwffff}
\mathbf
Z^{(k)}\mathbf p^{(k)}\leq \frac{1}{\beta^{(k)}}\mathbf p^{(k)}\ \text{for
all}\ k\in \mathbb N.
\end{equation}
 Moreover, there exists $1\leq i\leq N$ such that $[\mathbf
Z^{(k)}\mathbf p^{(k)}]_i = \frac{1}{\beta^{(k)}}[\mathbf
p^{(k)}]_i$.
We note that each vector $\mathbf p^{(k)}$ is a unit vector, as
$||\mathbf p^{(k)}||=1$. By the Bolzano-Weierstrass Theorem and the
compactness of the unit ball in $\mathbb R^{N}$, there exists a
convergent subsequence, that is, $\mathbf p^{(k_j)}\rightarrow \mathbf
p^{*}$. By Lemma
\ref{lemma:decreasing}, $\beta^{(k_j)}\rightarrow \beta^*$. Suppose at $\mathbf p^*$, $\mathbf Z^* \in \mathcal Z$ is one
of the matrices that satisfy $\mathbf Z^{*} \mathbf p^{*} \geq
\mathbf Z \mathbf p^{*}$ for all $\mathbf Z\in \mathcal Z$. Taking
the limit of \eqref{equ:pwffff} with respect to the subsequence indexed by
$k_j$, 
we have $\mathbf Z^{*}\mathbf p^{*}\leq \frac{1}{\beta^*} \mathbf
p^{*}$. 
 If $\mathbf Z^{*}\mathbf p^{*}= \frac{1}{\beta^*}  \mathbf p^{*}$, since $\mathbf Z^*$ is irreducible, $\beta^* =\frac{1}{\lambda(\mathbf Z^*)}\geq \frac{1}{\max_{\mathbf Z\in
   \mathcal Z(\boldsymbol\mu)} \{\lambda(\mathbf Z) \}}$. On the other hand, $\beta^*\leq \frac{1}{\max_{\mathbf Z\in
   \mathcal Z(\boldsymbol\mu)} \{\lambda(\mathbf Z) \}}$ by Lemma \ref{lemma:decreasing}. Therefore, $\beta^* = \frac{1}{\max_{\mathbf Z\in
   \mathcal Z(\boldsymbol\mu)} \{\lambda(\mathbf Z) \}}$.

If $\mathbf Z^{*}\mathbf p^{*}\neq\frac{1}{\beta^*}\mathbf p^{*}$,
since $\mathcal Z$ is primitive, there exists integer $n$ such that an
arbitrary product of $n$ matrices from $\mathcal Z$ is positive, i.e.,
$\Theta(n)>\mathbf 0$, and therefore $\Theta(n)\mathbf Z^{*}\mathbf
p^{*}< \Theta(n)\frac{1}{\beta^*}\mathbf p^{*}$. By the continuity
of the mapping, there exists $\mathbf p^{(k)}$ close enough to
$\mathbf p^{*}$ such that $\Theta(n)\mathbf Z^{*}\mathbf
p^{(k)}< \Theta(n)\frac{1}{\beta^*}\mathbf p^{(k)}$ and $\mathbf
Z^{(k)}=\mathbf Z^{*}$.
We now apply the algorithm for $n$ more iterations from $\mathbf
p^{(k)}$. For $i=0,1,\ldots,n-1$ we have the following inequalities 
\begin{equation}\label{equ:mfwff}
\mathbf Z^{(k+i+1)}\mathbf p^{(k+i)}  \leq \mathbf Z^{(k+i)}\mathbf p^{(k+i)} 
\end{equation}
due to that the selection matrix satisfies $\mathbf Z^{(k+i)}\mathbf p^{(k+i)} \geq \mathbf
Z\mathbf p^{(k+i)}$ for all $\mathbf Z\in \mathcal Z$. Meanwhile
by Algorithm \ref{iter},
\begin{align*}
\mathbf p^{(k+i)}  = \frac{\mathbf y^{(k+i-1)}}{||\mathbf
                    y^{(k+i-1)}||} = \frac{\mathbf Z^{(k+i-1)}\mathbf
                    p^{(k+i-1)}}{||\mathbf Z^{(k+i-1)}\mathbf p^{(k+i-1)}||} =  \frac{\mathbf Z^{(k+i-1)}\cdots\mathbf Z^{(k+1)}\mathbf Z^{(k)}\mathbf p^{(k)}}{||\mathbf Z^{(k+i-1)}\cdots\mathbf Z^{(k+1)}\mathbf Z^{(k)}\mathbf p^{(k)}||}.
\end{align*}
By substituting $\mathbf p^{(k+i)}$ into \eqref{equ:mfwff}, we get
\begin{equation}\label{equ:moaff}
\mathbf Z^{(k+i+1)}\mathbf Z^{(k+i-1)}\mathbf Z^{(k+i-2)} \cdots
\mathbf Z^{(k+1)}\mathbf Z^{(k)}\mathbf p^{(k)} \leq \mathbf Z^{(k+i)}\mathbf Z^{(k+i-1)}\mathbf Z^{(k+i-2)} \cdots
\mathbf Z^{(k+1)}\mathbf Z^{(k)}\mathbf p^{(k)}. 
\end{equation}
Let us take a look at these inequalities step by step. By
\eqref{equ:mfwff} for $i=0$,
$\mathbf Z^{(k+1)}\mathbf p^{(k)}  \leq \mathbf Z^{(k)}\mathbf
p^{(k)}$. By multiplying $\mathbf Z^{(k+2)}$ on both side of the inequality, we have
\begin{equation}\label{equ:moasdf}
\mathbf Z^{(k+2)}\mathbf Z^{(k+1)}\mathbf p^{(k)}  \leq \mathbf Z^{(k+2)}\mathbf Z^{(k)}\mathbf
p^{(k)}.
\end{equation}
By \eqref{equ:moaff} for $i=1$, $\mathbf Z^{(k+2)}\mathbf
Z^{(k)}\mathbf p^{(k)}  \leq \mathbf Z^{(k+1)}\mathbf Z^{(k)}\mathbf
p^{(k)}$. Along with \eqref{equ:moasdf}, we have 
\begin{equation*}
\mathbf Z^{(k+2)}\mathbf Z^{(k+1)}\mathbf p^{(k)}  \leq \mathbf Z^{(k+1)}\mathbf Z^{(k)}\mathbf
p^{(k)}.
\end{equation*}
By multiplying $\mathbf Z^{(k+3)}$ on both side of the above inequality, we have $\mathbf Z^{(k+3)}\mathbf Z^{(k+2)}\mathbf Z^{(k+1)}\mathbf p^{(k)}  \leq \mathbf Z^{(k+3)}\mathbf Z^{(k+1)}\mathbf Z^{(k)}\mathbf
p^{(k)}$. By \eqref{equ:moaff} for $i=2$, $\mathbf Z^{(k+3)}\mathbf Z^{(k+1)}\mathbf
Z^{(k)}\mathbf p^{(k)}  \leq \mathbf Z^{(k+2)}\mathbf Z^{(k+1)}\mathbf Z^{(k)}\mathbf
p^{(k)}$. So
\begin{equation*}
\mathbf Z^{(k+3)}\mathbf Z^{(k+2)}\mathbf Z^{(k+1)}\mathbf p^{(k)}  \leq \mathbf Z^{(k+2)}\mathbf Z^{(k+1)}\mathbf Z^{(k)}\mathbf
p^{(k)}.
\end{equation*}
By repeating this procedure for $n-1$ times, we can finally get
\begin{equation}\label{equ:vubii}
\mathbf Z^{(k+n)}\mathbf Z^{(k+n-1)}\cdots\mathbf Z^{(k+1)}\mathbf p^{(k)}\leq \mathbf Z^{(k+n-1)}\mathbf Z^{(k+n-2)}\cdots\mathbf Z^{(k+1)}\mathbf Z^{(k)}\mathbf p^{(k)}.
\end{equation}
Since $\Theta(n)\mathbf Z^{*}\mathbf
p^{(k)}< \Theta(n)\frac{1}{\beta^*}\mathbf p^{(k)}$ holds for
arbitrary $\Theta(n)$, we let $\Theta(n) =\mathbf Z^{(k+n)}\mathbf
Z^{(k+n-1)}\cdots\mathbf Z^{(k+1)} $. Along with \eqref{equ:vubii}, we have
\begin{align*}
 \mathbf Z^{(k+n)}\mathbf
Z^{(k+n-1)}\cdots\mathbf Z^{(k+1)}\mathbf Z^{*}\mathbf
p^{(k)}&= \Theta(n)\mathbf Z^{*}\mathbf
p^{(k)} \\
&< \Theta(n)\frac{1}{\beta^*}\mathbf p^{(k)}\\
&= \frac{1}{\beta^*}\mathbf Z^{(k+n)}\mathbf
Z^{(k+n-1)}\cdots\mathbf Z^{(k+1)}\mathbf p^{(k)}\\
&\leq \frac{1}{\beta^*}\mathbf Z^{(k+n-1)}\mathbf Z^{(k+n-2)}\cdots\mathbf Z^{(k+1)}\mathbf Z^{*}\mathbf p^{(k)}.
\end{align*}
By multiplying $\frac{1 }{||\mathbf
Z^{(k+n-1)}\cdots\mathbf Z^{(k+1)}\mathbf Z^{*}\mathbf
p^{(k)}||} $ on both side of the inequality, we have 
\begin{align*}
\mathbf Z^{(k+n)}\mathbf p^{(k+n)} &= \mathbf Z^{(k+n)}\frac{\mathbf
Z^{(k+n-1)}\cdots\mathbf Z^{(k+1)}\mathbf Z^{*}\mathbf
p^{(k)}}{||\mathbf
Z^{(k+n-1)}\cdots\mathbf Z^{(k+1)}\mathbf Z^{*}\mathbf
p^{(k)}||} \\
& <\frac{1}{\beta^*} \frac{\mathbf
Z^{(k+n-1)}\cdots\mathbf Z^{(k+1)}\mathbf Z^{*}\mathbf
p^{(k)}}{||\mathbf
Z^{(k+n-1)}\cdots\mathbf Z^{(k+1)}\mathbf Z^{*}\mathbf
p^{(k)}||} \\
&=\frac{1}{\beta^*}\mathbf p^{(k+n)}.
\end{align*}
This implies $\beta^{(k+n)}>\beta^{*}$, which contradicts with that
$\beta^{*}$ is the limit. Hence there must have $\mathbf Z^{*}\mathbf p^{*}=\frac{1}{\beta^*}\mathbf p^{*}$.

We prove $\mathbf p^{(k)} \rightarrow \mathbf p^*$ by
contradiction. Suppose there exists another subsequence such that $\mathbf
p^{k_j'}\rightarrow \mathbf p'$ and $\mathbf p^*\neq \mathbf
p'$. Then $\mathbf Z^*\mathbf
p' \leq \frac{1}{\beta^*}\mathbf p'$. Meanwhile we already
have $\mathbf Z^{*}\mathbf p^{*} = \frac{1}{\beta^*}\mathbf p^{*}$. By the Subinvariance
Theorem in \cite{92048} (pp. 23), $\mathbf p^*=\mathbf p'$, which
contradicts with the assumption that $\mathbf p^*\neq \mathbf
p'$. Therefore $\mathbf p^{(k)}$ converges to $\mathbf p^*$.

Since $\mathbf Z^{*}\mathbf p^{*} = \frac{1}{\beta^*}\mathbf p^{*}$
and $\mathbf Z\mathbf p^{*} \leq \frac{1}{\beta^*}\mathbf p^{*}$ for
all $\mathbf Z\in \mathcal Z(\boldsymbol\mu)$, it
is ready to see that $\lim_{\alpha\rightarrow
  \infty} \gamma_i(\alpha \mathbf p^*) = \min_{\mathbf Z\in \mathcal
  Z(\mathbf 1)} \{\frac{p^*_i}{[\mathbf Z\mathbf p^*]_i}
\}=\beta^*\mu_i$ for all $i=1,\ldots,N$.  So $\lim_{\alpha\rightarrow
  \infty} \Gamma(\alpha \mathbf p^*) = \beta^*\boldsymbol\mu$.
\end{proof}


\end{document}